\documentclass{article}

\usepackage{xspace}
\usepackage{amsmath, amssymb, latexsym}
\usepackage{bbm}
\usepackage{stmaryrd}
\usepackage{mathrsfs}
\usepackage{color}
\usepackage{dsfont}
\usepackage{microtype}
\usepackage{times}
\title{High-Quality Synthesis Against Stochastic Environments} 

 \newtheorem{theorem}{Theorem}[section]

\newtheorem{remark}{Remark}
\newtheorem{example}{Example}

 \newtheorem{lemma}[theorem]{Lemma}

 \def\squarebox#1{\hbox to #1{\hfill\vbox to #1{\vfill}}}
 
 \newcommand{\qed}{\hspace*{\fill}
 	\vbox{\hrule\hbox{\vrule\squarebox{.667em}\vrule}\hrule}\smallskip}
 \newenvironment{proof}{\begin{trivlist}
 		\item[\hspace{\labelsep}{\bf\noindent Proof: }]
 	}{\qed\end{trivlist}}
\newcommand{\shaull}[1]{}
\newcommand{\short}[1]{}
\newcommand{\set}[1]{{\{#1\}}}
\newcommand{\Nat}{\mathbbm{N}}

\newcommand{\zug}[1]{\langle #1  \rangle}

\newcommand{\val}{{\rm val}}

\newcommand{\sem}[1]{[\![#1]\!]} 

\mathchardef\mhyphen="2D
\newcommand{\True}{\mathtt{True}}
\newcommand{\False}{\mathtt{False}}
\newcommand{\LTL}{{\ensuremath{\rm LTL}}\xspace}

\newcommand{\Next}{\mathsf{X}}
\newcommand{\Ev}{\mathsf{F}}
\newcommand{\Alw}{\mathsf{G}}
\newcommand{\Until}{\mathsf{U}}

\newcommand{\T}{{\mathcal T}}

\newcommand{\DPW}{\mbox{\rm DPW}\xspace}
\newcommand{\DPWs}{\mbox{\rm DPWs}\xspace}

\newcommand{\stam}[1]{}

\newcommand{\B}{{\cal B}}

\newcommand{\A}{{\cal A}}
\newcommand{\M}{{\cal M}}
\newcommand{\F}{{\cal F}}
\newcommand{\EE}{\mathbb{E}}
\newcommand{\RR}{\mathbb{R}}

\newcommand{\Competence}{\triangledown}

\newcommand{\factorU}{\Competence}

\newcommand{\avg}[1]{\oplus_{#1}}

\renewcommand{\phi}{\varphi}

\newcommand{\maxs}[1]{\max\set{#1}}

\newcommand{\Inf}{{\rm inf}}

\newcommand{\ST}{ : \:}
\newcommand{\SecondST}{ \: | \:} 

\newcommand{\Fu}[2]{\ensuremath {{#1}{[#2]}}\xspace}

\newcommand{\FLTL}{\Fu{\LTL}{\F}}
\newcommand{\LTLF}{\FLTL}

\newcommand{\Act}{A}
\newcommand{\MDPProb}{{\rm P}}
\newcommand{\MDPre}{{\gamma}}

\newcommand{\tAP}{{2^{I\cup O}}}

\newcommand{\tAPo}{{{(2^{I\cup O})}^\omega}}
\newcommand{\tIN}{{2^I}}
\newcommand{\tOUT}{{2^O}}

\newcommand{\tINo}{{{(2^I)}^\omega}}
\newcommand{\tOUTo}{{{(2^O)}^\omega}}

\newcommand{\baderg}{{\mathscr{R}}}
\newcommand{\synprob}{{SHQSyn}\xspace}

\begin{document}





\author{Shaull Almagor and Orna Kupferman}
\date{The Hebrew University, Jerusalem, Israel.}

\maketitle

\begin{abstract}
In the classical synthesis problem, we are given a linear temporal logic (\LTL) formula $\psi$ over sets of input and output signals, and we synthesize a transducer that realizes $\psi$:  with every sequence of input signals, the transducer associates a sequence of output signals so that the generated computation satisfies $\psi$.
One weakness of automated synthesis in practice is that it pays no attention to the quality of the synthesized system. Indeed, the classical setting is Boolean: a computation satisfies a specification or does not satisfy it. Accordingly, while the synthesized system is correct, there is no guarantee about its quality. 
In recent years, researchers have considered extensions of the classical Boolean setting to a quantitative one. The logic $\FLTL$ is a multi-valued logic that augments LTL with quality operators. 
The satisfaction value of an $\FLTL$ formula is a real value in $[0,1]$, where the higher the value is, the higher is the quality in which the computation satisfies the specification. 

Decision problems for LTL become search or optimization problems for $\FLTL$. In particular, in the synthesis problem, the goal is to generate a transducer that satisfies the specification in the highest possible quality. 
Previous work considered the worst-case setting, where the goal is to maximize the quality of the computation with the minimal quality. We introduce and solve the stochastic setting, where the goal is to generate a transducer that maximizes the expected quality of a computation, subject to a given distribution of the input signals. Thus, rather than being hostile, the environment is assumed to be probabilistic, which corresponds to many realistic settings. 
We show that the problem is 2EXPTIME-complete, like classical LTL synthesis. The complexity stays 2EXPTIME also in two extensions we consider: one that maximizes the expected quality while guaranteeing that the minimal quality is, with probability $1$, above a given threshold, and one that allows assumptions on the environment.
\end{abstract}

\section{Introduction}
\label{sec:intro}
{\em Synthesis\/} is the automated construction of a system from its specification: given a linear temporal logic (LTL)  formula $\psi$ over sets $I$ and $O$ of input and output signals, we synthesize a finite-state system that {\em realizes\/} $\psi$ \cite{Chu63,PR89a}. At each moment in time, the system reads a truth assignment, generated by the environment, to the signals in $I$, and it generates a truth assignment to the signals in $O$. Thus, with every sequence of inputs, the system associates a sequence of outputs. 
The system realizes $\psi$ if all the computations that are generated by the interaction satisfy $\psi$. 

One weakness of automated synthesis in practice is that it pays no attention to the quality of the synthesized system. Indeed, the classical setting is Boolean: a computation satisfies a specification or does not satisfy it. Accordingly, while the synthesized system is correct, there is no guarantee about its quality. This is a crucial drawback, as designers would be willing to give-up manual design only if automated-synthesis algorithms return systems of comparable quality. 
In recent years, researchers have considered several extensions and variants of the classical setting of synthesis.
One class of extensions stays in the Boolean setting. For example, in practice we can often make assumptions on the behavior of the environment. An assumption may be direct, say given by an LTL formula that restricts the set of possible sequences of inputs \cite{CHJ08}, or conceptual, say rationality from the side of the environment, which may have its own objectives \cite{FKL10}, or a bound on the size of the environment and/or the generated system \cite{SF07,KLVY11}. Another class of extensions moves to a quantitative setting, where a specification may have different satisfaction values in different systems. For example, in \cite{BCHJ09}, the input to the synthesis problem includes also Mealy machines that grade different realizing systems. As another example, in~\cite{ABK13}, the specification formalism is the multi-valued logic \FLTL, which augments LTL with quality operators. The satisfaction value of an $\FLTL$ formula is a real value in $[0,1]$, where the higher the value is, the higher is the quality in which the computation satisfies the specification. The synthesis algorithm then seeks systems of the highest possible quality. A quantitative approach can be taken also with Boolean specifications and involves a probabilistic view: the environment is assumed to generate input sequences according to some probability distribution. Then, instead of requiring the system to satisfy the specification in all computations generated by the environment, we measure the probability with which this happens~\cite{KP13}.

Combining the multi-valued approach with the probabilistic one has led to the use of Markov Decision Processes (MDPs). Indeed, MDPs are a clean mathematical model that allows the analysis of quantitative objectives in a probabilistic environment. The intricacy of MDPs has led, in turn, to a plethora of works on synthesis with various constraints and reward models (e.g.~\cite{BKKM14,CD11,CHJ04,CKK15,CR15}). The starting point of these works  is the MDPs. This is puzzling,  as while MDPs offer a very clean framework for the analysis, they do not serve as a specification formalism. Thus, the crucial step of actually obtaining the MDPs is missing.
\stam{Naturally, researchers seek to combine the two quantitative approaches. That is, realizing a quantitative specification against a probabilistic environment. This search led to the use of Markov Decision Processes (MDPs) as a very successful paradigm. Indeed, MDPs are a clean mathematical model that allows the analysis of quantitative rewards and a probabilistic environment. The intricacy of MDPs led, in turn, to a plethora of works on synthesis with various constraints and reward models (e.g.~\cite{BKKM14,CD11,CHJ04,CKK15,CR15}). However, while MDPs serve well as a technical tool for algorithmic use, they are unfit as a specification formalism for practical use, much like automata serve as a technical tool in the Boolean setting, but are not typically used as a specification formalism. Thus, quantitative specification formalisms need to catch up to the existing technical framework before these methods can be applied in practice.
}

In this work, we consider {\em stochastic high-quality synthesis}, 
which combines the multi-valued approach with the probabilistic one. We build on known techniques for MDPs, and still keep the specification formalism accessible to designers. 
The specification is given by an \FLTL formula, the environment is assumed to be probabilistic, and we seek a system that maximizes the expected satisfaction value. To explain the setting better, let us first review shortly \LTLF. The linear temporal logic \FLTL extends LTL with an arbitrary set $\F$ of functions over $[0,1]$. Using the functions in $\F$, a specifier can formally and easily prioritize the different ways of satisfaction. 
The logic \FLTL is really a family of logics, each parameterized by a set  
$\F\subseteq \{f:[0,1]^k\to [0,1] \SecondST k\in\Nat\}$ of functions (of arbitrary arity) over $[0,1]$.
For example, as in earlier work on multi-valued extensions of LTL (c.f., \cite{FLS08}), the set $\F$ may contain the $\min\set{x,y}$, $\max\set{x,y}$, and $1-x$ functions, which are the standard quantitative analogues of the $\wedge$, $\vee$, and $\neg$ operators. The novelty of \FLTL is the ability to manipulate values by arbitrary functions. For example, $\F$ may contain the quantitative operator $\factorU_\lambda$, for $\lambda \in [0,1]$, which tunes down 
the quality of a sub-specification. Formally, the satisfaction value of the specification $\factorU_\lambda \varphi$ is the multiplication of the satisfaction value of $\varphi$ by $\lambda$.  Another useful operator is the weighted-average function $\avg{\lambda}$. There, the satisfaction value of the formula $\phi\avg{\lambda}\psi$ is the weighted (according to $\lambda$) average between the satisfaction values of $\phi$ and $\psi$. This enables the quality of the system to be an interpolation of different aspects of it. 
As an example, consider the \FLTL formula $\phi=\Alw({\it req} \rightarrow ({\it grant} \avg{\frac{2}{3}} \Next {\it grant}))$. The formula 
 specifies the fact that we want requests to be granted immediately and the grant to hold for two transactions. When this always holds, the satisfaction value is $\frac{2}{3}+\frac{1}{3}=1$. We are quite okay with grants that are given immediately and last for only one transaction, in which case the satisfaction value is $\frac{2}{3}$, and less content when grants arrive with a delay, in which case the satisfaction value is $\frac{1}{3}$.
 
Consider a system that receives requests and generates grants and consider a specification $\psi$ that have $\phi$ above as a sub-formula. Other sub-formulas of $\psi$ may require the system to generate as few grants as possible, say with $\phi'=(\Ev \Alw (\neg {\it req})) \rightarrow (\Alw \neg ({\it grant} \wedge \Next {\it grant}))$. That is, if requests eventually stop arriving, then there cannot be two successive grants. The specification $\psi$ cannot be realized with satisfaction value $1$, as the system does not know in advance whether requests eventually stops arriving. Therefore, in order to get a satisfaction value above $0$ in the subformula $\phi'$, the system must not generate two successive grants, bounding the satisfaction value of the subformula $\phi$ by $\frac23$. If, however, the input signals are distributed so that ${\it req}$ may hold with a positive probability at each moment in time, then the probability that an input sequence satisfies $\Ev \Alw (\neg {\it req})$ is $0$, causing $\phi'$ to be satisfied (that is, to have satisfaction value $1$) with probability $1$. Accordingly, under this assumption, a system that grants requests immediately and for two transactions has expected satisfaction value $1$.

Formally, one can measure the quality of a system ${\cal S}$ with respect to an \LTLF specification taking three approaches. 
In the {\em worst-case approach}, the environment is assumed to be hostile and we care for the minimal satisfaction value of some computation of ${\cal S}$. In the {\em almost-sure approach}, the environment is assumed to be stochastic 
and we care for the maximal satisfaction value that is generated with probability $1$. Then, in the 
{\em stochastic approach}, the environment is assumed to be stochastic and we care for the expected satisfaction value of the computations of ${\cal S}$, assuming some given distribution on the inputs sequences. 

\begin{example}
\label{xpl:intro hard drive 1}
{\rm 
Consider a battery-replacement controller for a certain hardware.
A computation of the hardware lasts $k$ steps. Some steps during the execution are {\em stations}, in which the battery can be replaced.  For example, the hardware may be an electric car whose battery can only be replaced at charging stations. The controller should decide at which stations it replaces the battery. On the one hand, it is wasteful to replace the battery early. On the other hand, the occurrence of stations is random, and the controller does not know whether stations are going to be encountered in the future.

Since it is wasteful to replace the battery early, the specification states that replacing it in step $1\le t\le k$ lowers the satisfaction value to $t/k$. Missing, however, all stations incurs satisfaction value $0$. We assume that each step is a station with probability $p\in [0,1]$. 

Formally, the specification for the controller is over the sets $I=\set{{\it station}}$ and $O=\set{{\it replace}}$, and is a conjunction $\phi_1\wedge\phi_2\wedge \phi_3$ of three $\LTLF$ formulas  (the abbreviation $\Next^i$ stands for a sequence of $i$ nested $\Next$ (next) operators):
\begin{itemize}
	\item $\phi_1=\Alw({\it replace}\to {\it station})$, which requires that we only replace the battery in stations,
	\item $\phi_2=(\bigvee_{1 \leq t \leq k} \Next^k {\it station})  \to  (\bigvee_{1 \leq t \leq k} \Next^k {\it replace})$,  which states that the requirement to replace the battery needs to be satisfied only if at least one station has been encountered.
	\item $\phi_3=\bigwedge_{1 \leq t \leq k} \Next^t (\neg {\it replace} \vee \factorU_{\frac{t}{k}} {\it replace})$, which lowers the satisfaction value to $\frac{t_0}{k}$, for the minimal step $1\le t_0\le k$ in which the battery is replaced.
\end{itemize}

In order to ensure a positive satisfaction value in the worst case, a transducer must replace the battery on the first station it encounters. Such a transducer guarantees a satisfaction value of $\frac{1}{k}$, but has expected satisfaction value of $(1-p)^k(1-\frac1k)+\frac1k$, which tends to $0$ as $k$ increases. 

Trading-off the satisfaction value in the worst case for a higher expected satisfaction value, a controller may also replace the battery in later stations. For example, a transducer that replaces the battery only in the $k$-th step (if it is a station) has expected satisfaction value $(1-p)^k+p$. However, its satisfaction value in the worst case, in fact in $(1-p)$ of the computations, is $0$. 

In Appendix~\ref{apx:example1} we analyze the expected satisfaction value of a transducer that replaces the battery in the first station after position $t$, for $1\le t\le k$, and show, for example, that a transducer that replaces the battery starting in position $\frac{k}{2}$ has an expected satisfaction value that tends to $\frac12$ as $k\to\infty$, for every fixed $p\in (0,1)$.
\hfill \qed}
\end{example}

The worst case approach has been studied in \cite{ABK13}, where it is shown how to synthesize, given $\phi$, a system with a maximal worst-case satisfaction value.
In this paper, we consider the two other approaches.
We model a reactive system with sets $I$ and $O$ of input and output signals, respectively, by an $I/O$-transducer: a finite-state machine whose transitions are labeled by truth assignments to the signals in $I$ and whose states are labeled by  truth assignments to the signals in $O$.
We define and solve the {\em stochastic high-quality synthesis\/} problem (\synprob, for short). The input to the problem is  an $\LTLF$ formula $\phi$ over $I \cup O$, and we seek an $I/O$-transducer that maximizes the expected satisfaction value of a computation, under a given distribution of the inputs. 
We show that the maximal expected satisfaction value is always attained by a finite-state transducer, and that computing such a transducer takes time that is doubly-exponential in $\phi$, thus the problem is not more complex than the synthesis problem for LTL.

We continue to study two extensions of the \synprob problem. In the first extension, we add a lower bound on the satisfaction value that should be attained almost surely. Formally, the input to the {\em \synprob with threshold\/} problem is  an \FLTL formula $\phi$ and a threshold $t \in [0,1]$, and we seek a transducer that maximizes the expected satisfaction value of $\phi$, but such that the satisfaction value of $\phi$ in all its computations is at least $t$ with probability $1$. As we show, adding this restriction may lower the expected value. Also, our solution to the \synprob with threshold problem generalizes  high-quality synthesis in the almost-sure approach, which we solve too. 
This approach has been studied for MDPs in~\cite{CKK15,CR15}. We show that while we can readily apply the existing solutions, the fact that our original specification is an \FLTL formula allows us to obtain slightly better solutions, with simpler analysis.

The second extension is the quantitative analogue of 
synthesis with environment assumptions. As discussed above, adding assumptions on the environment is a useful extension in the Boolean setting \cite{CHJ08,LDS11}. In the  {\em \synprob with environment assumption\/} problem we get as input an \FLTL formula $\phi$ and an environment assumption $\psi$, given by means of an $\LTL$ formula, and we seek a transducer that maximizes the expected satisfaction value of $\phi$ in computations that satisfy $\psi$. We note that the ability to reason about the quality of satisfaction in the presence of environment assumptions suggests a quantitative solution to challenges that appear already in the Boolean setting. For example, in \cite{BEK15}, the authors study the annoying phenomenon of systems realizing a specification by causing the assumption to fail. They suggest a synthesis algorithm that increases the cooperation between the system and its environment. Using \LTLF, we can associate such a cooperation with high quality. 
We show that both extensions, of threshold and assumptions, as well as their combination, do not increase the complexity of the synthesis problem. 

From a technical perspective, solving the Boolean synthesis problem amounts to translating an \LTL formula to a deterministic parity automaton (\DPW), viewing this automaton as a two-player parity game in which the system plays against the environment, and finding a winning strategy for the system. When the environment is assumed to be stochastic, the two-player game becomes a Markov decision process (MDP) with a parity objective. Such MDPs were extensively studied in~\cite{CD11,CHJ04}. In order to handle the quantitative satisfaction values of $\FLTL$, we translate an \FLTL formula $\phi$ to a set of \DPWs associated with the different possible satisfaction values of $\phi$. From the latter we obtain a mean-payoff MDP. We show that a transducer that attains the maximal expected satisfaction value is embodied in this MDP, and can be found in polynomial time. The analysis of the MDP is based on a search for {\em controllably win recurrent\/} states \cite{CHJ04}. Adding a threshold $t \in [0,1]$, the strategies of the MDP are restricted to those that guarantee that the computation reaches, with probability $1$, end components that correspond to accepting runs of \DPWs associated with satisfaction values above $t$. 

Finally, in order to handle environment assumptions, we need to maximize the conditional expected satisfaction value, given the assumption. Maximizing conditional expectation is notoriously difficult, as, unlike unconditional expectation, it is not a linear objective. Thus, it is not susceptible to linear optimization techniques, which are the standard approach to find maximizing strategies in MDPs.
In our solution, we compose the MDP with the DPW for the assumption, which enables us to adopt techniques used in the context of conditional probabilities in MDPs \cite{BKKM14}. Intuitively, we add to the MDP transitions that ``redistributes'' the probability of computations that do not satisfy the assumption.  In both cases, the size of the analyzed MDP stays doubly exponential in $\phi$ (and the assumption, in the latter case), and the required transducer is embodied in it.

\stam{
The paper is organized as follows. In Section~\ref{sec:prelim} we define automata, transducers, and the probabilistic setting, as well as the linear temporal logic \LTLF. In Section~\ref{hqs}, we define the stochastic high-quality synthesis problem, give some motivating examples, and describe the main technical tool we are going to use for the solution, namely the MDP that embodies the solution transducer. Then, Section~\ref{sec:Expected synthesis problem} includes the solution for the problem and its complexity analysis, and Sections~\ref{sec:Expected synthesis with threshold} and~\ref{sec:assume-guarantee} describe the threshold and assumption extensions, respectively. Finally, Section~\ref{sec:extensions} discusses some further extensions. 
}
Due to lack of space, most proofs appear in the appendix.

\section{Preliminaries}
\label{sec:prelim}
\subsection{Automata and Transducers}

A (deterministic) {\em pre-automaton} is a tuple $\zug{\Sigma,Q,q_0,\delta}$, where $\Sigma$ is a finite alphabet, $Q$ is a finite set of states, $q_0 \in Q$ is an initial state, and 
$\delta: Q\times \Sigma \nrightarrow Q$ is a (partial) transition function. 
A run of the pre-automaton on a word $w=\sigma_1 \cdot \sigma_2 \cdots\in \Sigma^\omega$ is a sequence of states $q_0,q_1,q_2,\ldots$ such that $q_{j+1} = \delta(q_j,\sigma_{j+1})$ for all $j \geq 0$. Note that since $\delta$ is deterministic, the pre-automaton has 
at most 
one run on each word.

A {\em deterministic parity automaton\/} (\DPW, for short) is 
$\A=\zug{\Sigma,Q,q_0,\delta,\alpha}$, where $\zug{\Sigma,Q,q_0,\delta}$ is a pre-automaton, 
$\delta$ is a total function, 
and $\alpha:Q\to \set{1,...,d}$ is an acceptance condition that maps states to ranks. The maximal rank $d$ is the {\em index\/} of $\A$.
For a run $r=q_0,q_1,q_2,\ldots$ of $\A$, let 
$\Inf(r)$ be the set of states that occur in $r$ infinitely often. Formally, $\Inf(r)=\{q: q_j =q$ for infinitely many $j \geq 0\}$. The run $r$ is accepting if the maximal rank of a state in $\Inf(r)$ is even. Formally, $\max_{q \in \Inf(r)} \{\alpha(q)\}$ is even.
A word $w \in \Sigma^\omega$ is accepted by $\A$ if the run of $\A$ on $w$ is accepting. The language of $\A$, denoted $L(\A)$, is the set of words that $\A$ accepts. 

For finite sets $I$ and $O$ of input and output signals, respectively, an {\em $I/O$  transducer} is $\T=\zug{I,O,Q,q_0,\delta,\mu}$, where $\zug{2^I,Q,q_0,\delta}$ is a pre-automaton, and $\mu:Q\to \tOUT$ is a labeling function on the states. Intuitively, $\T$ models the interaction of an environment that generates in each moment in time a letter in $2^I$ with a system that responds with letters in $2^O$. Consider an input word $w=i_0 \cdot i_1 \cdots \in \tINo$ and let $q_0,q_1,\ldots$ be the run of $\T$ on $w$. The {\em output} of $\T$ on $w$ is then $o_1,o_2,\ldots\in \tOUTo$, where $o_j=\mu(q_{j})$ for all $j\ge 1$. Note that the first output assignment is that of $q_1$, thus $\mu(q_0)$ is ignored. This reflects the fact that the environment initiates the interaction. The {\em computation of $\T$ on $w$\/} is then 
$\T(w)=i_0\cup o_1,i_1\cup o_2,\ldots \in (2^{I \cup O})^\omega$.

\subsection{Markov Chains and Markov Decision Processes} 
A {\em Markov chain} (MC, for short) $\M=\zug{S,s_0,P}$ consists of a finite or countably-infinite state space $S$, an initial state $s_0 \in S$, and a stochastic transition function $P:S \times S \rightarrow [0,1]$. That is, for all $s \in S$, we have $\sum_{s' \in S} P(s,s') =1$. Intuitively, when a run of $\M$ is in state $s$, then it moves to state $s'$ with probability $P(s,s')$. A {\em run} of $\M$ is a finite or infinite sequence $s_0,s_1,s_2,...$ of states that starts in $s_0$.
The MC $\M$ induces a probability space on finite runs. Consider a finite run $r=s_0,s_1,...,s_k$. We define $\Pr(r)=\prod_{i=1}^{k-1}P(s_i,s_{i+1})$. Thus, the probability of a finite run is the product of the probabilities of its transitions. Let $\text{Cone}(r)$ be the set of all infinite runs that start with $r$. 
The MC $\M$ induces a probability space over the set of infinite runs of $\M$ that are generated by the cylinder sets $\text{Cone}(r)$, for finite runs $r$. Formally, for every $r\in S^*$, we have $\Pr(\text{Cone}(r))=\Pr(r)$.

An {\em ergodic component} of $\M$ is a strongly connected component of $\M$ from which no other component is reachable. Formally, it is a set $C\subseteq S$ such that for every $s,t\in C$ there exist a path $s_1,s_2,...,s_k$ of states in $C$ such that $s_1=s$, $s_k=t$, and $P(s_j,s_{j+1})>0$ for every $1\le j\le k$. In addition, for every $s\in C$ and $t\notin C$, it holds that $P(s,t)=0$. Let ${\cal C}$ be the set of maximal (w.r.t. containment) ergodic components of $\M$.
We associate with $\M$ an {\em ergodic reachability probability} $\rho:{\cal C}\to [0,1]$ such that $\rho(C)$ is the probability that a run of $\M$ reaches (and therefore remains forever in) $C$.

A {\em Markov decision process} (MDP) is 
$\M=\zug{S,s_0,(\Act_s)_{s\in S},\MDPProb,\MDPre}$, where $S$ is a finite set of states, $s_0\in S$ is an initial state, and $\Act_s$ is a finite set of actions that are available in state $s\in S$. Let $\Act=\bigcup_{s\in S}\Act_s$. Then,  $\MDPProb:S\times \Act\times S\nrightarrow [0,1]$ is a (partial) stochastic transition function: for every two states $s,s'\in S$ and action $a\in A_s$, we have that $\MDPProb(s,a,s')$ is the probability of moving from $s$ to $s'$ when action $a$ is taken. Accordingly, for every $s \in S$ and $a \in A_s$, we have $\sum_{s'\in S} \MDPProb(s,a,s')=1$. Finally, $\MDPre:S\to \RR$ is a reward function on the states. 

An MDP can be thought of as a game between a player, who chooses the action to be taken in each state, and nature, which stochastically chooses the next state according to the transition probabilities. The goal of the player is to maximize the average reward along the generated run in the MDP. We now formalize this intuition.

A {\em strategy} for the player in an MDP $\M$ (a strategy for $\M$, in short) is a function $f:S^+ \to A$ that suggests to the player an action to be taken given the history of the game so far. The strategy should suggest an available action, thus $f(s_0,\ldots,s_n)\in \Act_{s_n}$. A strategy is {\em memoryless} if it depends only on the current state. We can describe a memoryless strategy by $f:S\to A$, where again, $f(s) \in A_s$. 

Given a strategy $f$, we can obtain from $\M$ the MC $\M_f=\zug{S^+,s_0,P_f}$ in which the choice of actions is resolved according to $f$. Formally, if $u,u' \in S^+$ are such that there are $t \in S^*$ and $s,s' \in S$ such that $u=t \cdot s$ and $u'=t \cdot s \cdot s'$, then $P_f(u,u')= \MDPProb(s,f(t \cdot s),s')$. Otherwise, $P_f(u,u')=0$.  Note that $\M_f$ has an infinite state space.
If $f$ is memoryless, we can simplify the construction, and define $\M_f=\zug{S,s_0,P_f}$ with $P_f(s,s')=\MDPProb(s,f(s),s')$.

An {\em end component\/} in an MDP $\M$ is a set $C\subseteq S$ such that there exist action sets $(B_s)_{s\in S}$ with $B_s\subseteq A_s$ for every $s\in S$, and for every $s,t\in C$, there exists a path $s_1,s_2,...,s_k$ of states in $C$ such that $s_1=s$, $s_k=t$ and there exist actions $a_1,...,a_{k-1}$ such that $\MDPProb(s_j,a_i,s_{j+1})>0$ and $a_i\in B_{s_j}$ for every $1\le j\le k$. In addition, for every $s\in C$ and $a\in B_s$ it holds that $\sum_{t\in C}\MDPProb(s,a,t)=1$. Intuitively, an end component is a strongly-connected component in the MDP graph that nature cannot force to leave. Equivalently, $\M$ has a strategy to stay forever in $C$. Indeed, it is not hard to see that $C$ is an end component iff there is some strategy $f$ for $\M$ such that $C$ is an ergodic component of $\M_f$.

The {\em value} $\val_\M(f)$ 
(we omit the subscript when $\M$ is clear from context) 
of a strategy $f$ for $\M$ is the expected average reward of an infinite run in $\M_f$. Formally, for a run $r=s_0,s_1,s_2,\ldots$ of $\M_f$, we define $\MDPre(r)=\liminf_{m\to\infty}\frac{1}{m}\sum_{j=0}^{m}\MDPre(s_j)$, 
where for a state $s\in S^+$ of $\M_f$, the cost $\gamma(s)$ is induced by the last state of $\M$ in $s$.
In the stochastic setting, we view each sequence of inputs, and hence also each run $r$ and the reward on $r$, as a random variable. The {\em expected value\/} of a random variable is, intuitively, its average value, weighted by probabilities. 
Let $R_{\M,f}$ be the random variable whose value is the reward on runs in $\M_f$. We define 
$\val_{\M}(f)=\EE[R_{\M,f}]$. 
The {\em value} $\val(\M)$ of an MDP $\M$ is the maximal value of a strategy in $\M$.  
It is well known (see e.g. \cite{FV96}) that  $\val(\M)$ can be attained by a memoryless strategy, which can be computed in polynomial time.

For technical reasons, we sometimes use variants of MDPs. A {\em pre-MDP} is an MDP with no reward function. A {\em parity MDP} is a pre-MDP with a parity acceptance condition $\alpha:S\to\set{1,...,d}$. In a parity MDP, the goal of the player is to maximize the probability that the generated run satisfies the parity condition. Parity-MDPs were extensively studied in e.g. \cite{CHJ04a}.

\subsection{The logic $\LTLF$}
The logic $\LTLF$ is a multi-valued logic that extends the linear temporal logic \LTL with an arbitrary set of functions $\F\subseteq\set{f:[0,1]^k\to [0,1]:k\in \Nat}$ called {quality operators}. For example, $\F$ may contain the maximum or minimum between the satisfaction values of subformulas, their product, and
their average. This enables the specifier to refine the Boolean correctness notion and associate different possible ways of satisfaction with different truth values \cite{ABK13}.

Let $AP$ be a set of Boolean atomic propositions and let $\F$ be a set of function as described above. An $\LTLF$ formula is one of the following:
\begin{itemize}
\item $\True$, $\False$, or $p$, for $p\in AP$.
\item $f(\varphi_1,...,\varphi_k)$, $\Next\varphi_1$, or $\varphi_1\Until \varphi_2$, for \FLTL formulas $\varphi_1,\ldots,\varphi_k$ and a function $f\in \F$.
\end{itemize}
The semantics of \FLTL formulas is defined with respect to infinite computations over $\tAP$. For a computation $\pi=\pi_0,\pi_1,\ldots \in \tAPo$ and position $j \geq 0$, we use $\pi^j$ to denote the suffix $\pi_j,\pi_{j+1},\ldots$. The semantics maps a computation $\pi$ and an \FLTL formula $\phi$ to the {\em satisfaction value\/} of $\varphi$ in $\pi$, denoted $\sem{\pi,\varphi}$. The satisfaction value is in $[0,1]$ and is defined inductively as described in Table~\ref{tab:semantics} below.
\begin{table}[htb!]
\centering
\begin{tabular}{|c|c||c|c|}
\hline
Formula  & Satisfaction value & Formula  & Satisfaction value\\
\hline\hline
$\sem{\pi,\True}$  & 1 & $\sem{\pi,f(\varphi_1,...,\varphi_k)}$ & $f(\sem{\pi,\varphi_1},...,\sem{\pi,\varphi_k})$\\
\hline
$\sem{\pi,\False}$  & 0 & $\sem{\pi, \Next \varphi_1}$ & $\sem{\pi^1, \varphi_1}$\\
\hline
$\sem{\pi, p}$ &$\begin{array}{ll}
1 & \mbox{ if }p\in \pi_0\\
0 & \mbox{ if }p\notin \pi_0\\
\end{array}$ & $\sem{\pi, \varphi_1 \Until \varphi_2}$ & $\max\limits_{0\leq i < |\pi|} \{ \min \{\sem{\pi^i,\varphi_2},  \min\limits_{0\leq j < i}\sem{\pi^j,\varphi_1} \} \}$\\
\hline 
\end{tabular}
\caption{The semantics of \FLTL.} 
\label{tab:semantics}
\end{table}

\vspace{-.5cm}
The logic \LTL can be viewed as \FLTL for $\F$ that models the usual Boolean operators. For simplicity, we use the common such functions as abbreviation, as described below. In addition, we introduce notations for two useful quality operators, namely factoring and weighted average. Let $x,y,\lambda \in[0,1]$. Then, 

\stam{
\begin{itemize}
\item $\neg x=1-x$, 
\item $x\vee y=\maxs{x,y}$, 
\item $x\wedge y=\min\set{x,y}$, 
\item $x\to y = \maxs{1-x,y}$.
\item For a parameter $\lambda \in [0,1]$, we have 
\begin{itemize}
\item $\factorU_\lambda x=\lambda\cdot x$,
\item $x\avg{\lambda}y=\lambda \cdot x+(1-\lambda)\cdot y$.
\end{itemize}
\end{itemize}
}

\begin{center}
\begin{tabular}{lllll}
$\bullet ~~ \neg x=1-x$  &~~~& $\bullet ~~ x\vee y=\maxs{x,y}$ &~~~& $\bullet ~~ x\wedge y=\min\set{x,y}$ \\
$\bullet ~~ x\rightarrow y=\maxs{1-x,y}$ &~~& $\bullet ~~ \factorU_\lambda x=\lambda\cdot x$&~~&$\bullet ~~ x\avg{\lambda}y=\lambda \cdot x+(1-\lambda)\cdot y$
\end{tabular}
\end{center}

\begin{example}
\label{xpl: fltl formula}
{\rm Consider a scheduler that receives requests and generates grants and consider the \FLTL formula $\phi=\phi_1\wedge \phi_2$, with $\phi_1=\Alw({\it req} \rightarrow \Next({\it grant}\avg{\frac{2}{3}}\Next{\it grant}))$ and $\phi_2=\neg (\factorU_{\frac34}\Alw \neg{\it req})$. The satisfaction value of the formula $\phi_1$ is $1$ if every request is granted in the next cycle and the grant lasts for two consecutive cycles. If the grant lasts for only one cycle, then the satisfaction value is reduced to $\frac{2}{3}$ if it is the cycle right after the request, and to $\frac{1}{3}$ if it is the next one. In addition, the conjunction with $\phi_2$ implies that if there are no requests, then the satisfaction value is at most $\frac14$. The example demonstrates how \FLTL can conveniently prioritize different scenarios, as well as embody vacuity considerations in the formula. \hfill \qed
}
\end{example}

For an \FLTL formula $\varphi$, let $V(\varphi)=\{\sem{\pi,\varphi} \ST \pi\in (2^{AP})^{\omega}\}$. That is, $V(\varphi)$ is the set of possible satisfaction values of $\varphi$ in arbitrary computations. 
\begin{theorem}{\rm \cite{ABK13}} 
\label{from abk}
Consider an \FLTL formula $\varphi$.
\begin{itemize}
\item
$|V(\varphi)|\le 2^{|\varphi|}$.
\item
For every predicate $\theta \subseteq [0,1]$, there exists a DPW $\A_{\varphi,\theta}$ such that $L(\A_{\varphi,\theta})=\{ \pi :\sem{\pi,\varphi} \in \theta\}$. Furthermore, $\A_{\varphi,\theta}$ has at most $2^{2^{O(|\varphi|)}}$ states and its index is at most $2^{|\varphi|}$.
\end{itemize}
\end{theorem}
 
\section{High-Quality Synthesis}
\label{hqs}

Consider an $I/O$-transducer $\T$ and an \FLTL formula $\phi$ over $I\cup O$. 
Each computation of $\T$ may have a different satisfaction value for $\phi$. 
We can measure the quality of $\T$ taking three approaches: 
\begin{itemize}
\item
{\em Worst-case approach:} The environment is assumed to be hostile and we care for the minimal satisfaction value of some computation of $\T$. Formally, $\sem{\T,\phi}_w=\min\{\sem{\T(w),\phi} : w \in (2^I)^\omega\}$. Note that no matter what the input sequence is, the specification $\phi$ is satisfied with value at least 
$\sem{\T,\phi}_w$.
\item
{\em Almost-sure approach:} The environment is assumed to be stochastic 
and we care for the maximal satisfaction value that is generated with probability $1$. Formally, given a distribution $\nu$ of $(2^I)^\omega$, we define 
$\sem{\T,\phi}^\nu_a=\max\{v: \text{ there is } W \text{ with } \nu(W)=1\text{ and }\sem{T(w),\phi}\ge v \text{ for every }w\in W\}$.
Note that the specification $\phi$ is satisfied almost surely with value at least $\sem{\T,\phi}^\nu_a$.
\item
{\em Stochastic approach:} The environment is assumed to be stochastic and we care for the expected satisfaction value of the computations of $\T$, assuming some given distribution on the inputs sequences. Formally, let $X_{\T,\phi}: (2^I)^\omega \rightarrow \RR$ be a random variable that assigns each sequence $w \in (2^I)^\omega$ of input signals with $\sem{\T(w),\phi}$. Then, given a distribution $\nu$ of $(2^I)^\omega$, we define $\sem{\T,\phi}_s^\nu=\EE[X_{\T,\phi}]$, when the sequences in $(2^I)^\omega$ are sampled according to $\nu$. 
\end{itemize}
The worst case approach has been studied in \cite{ABK13}, where it is shown how to find $\sem{\T,\phi}_w$ and how to synthesize, given $\phi$, a transducer with a maximal worst-case satisfaction value.
In this paper, we consider the stochastic approach. For simplicity, we consider environments with a uniform distribution on the input signals. That is, $\nu$ is such that in each moment in time, each input signal holds with probability $\frac{1}{2}$, thus the probability of each letter in $2^I$ is $\frac{1}{2^{|I|}}$ (see Remark~\ref{rmrk:uniform distribution}). Since $\nu$ is fixed, we omit it from the notation and use $\sem{\T,\phi}_a$ and $\sem{\T,\phi}_s$.

\begin{remark}
	\label{rmrk:uniform distribution}
{\rm	{\bf [On the choice of a uniform distribution]} Recall that we consider a uniform distribution on the letters in $\tIN$. In practice, the distribution on the truth assignments to the input signals may be richer. In the general case, such a distribution can be given by an MDP. 
	
	Adjusting our setting and algorithms to handle such distributions involves only a small technical elaboration, orthogonal to the technical challenges that exist already in the setting of a uniform distribution. Accordingly, throughout the paper we assume a uniform distribution. 
	In Section~\ref{nud}, we describe how our setting and algorithms are extended to the general case.}
	\hfill \qed
\end{remark}

\begin{example}
\label{xpl:hard drive 1}
{\rm 
Consider a hard-drive writing protocol that
needs to finalize a write operation through some connection. The connection needs to be closed as soon as possible, to allow access to the drive.  However, data may still arrive in the first two cycles, and if the connection is closed in the first cycle, then the data that arrives in the second cycle gets lost. The issue is that the decision as to whether to close the connection is made during the first cycle, before the protocol knows whether data is going to arrive in the second cycle.
The specification that formulates the above scenario is over $I=\{data\}$ and $O=\{close\}$ and is
$\phi =((\Next data)\to \neg close)\wedge ((\neg \Next data) \to close) \vee \factorU_\frac12\Next close).$

That is, if data arrives in the second cycle, then we should not close the connection in the first cycle. In addition, if data does not arrive in the second cycle, we should 
close the connection in the first cycle -- this would give us satisfaction value $1$ in the second conjunct, but we may also close the connection only in the second cycle, which would guarantee a satisfaction value of $1$ in the  first conjunct, but would reduce the satisfaction value of the second conjunct to $\frac{1}{2}$ in cases data does not arrive in the second cycle.  

Let $p \in [0,1]$ be the probability that data arrives in the second cycle. 
Consider a transducer $\T_1$ that closes the connection in the first cycle. With probability $p$, we have that $\Next data$ holds, in which case $\phi$ has satisfaction value $0$. Also, with probability $1-p$, we have that $\Next data$ does not hold and the satisfaction value of $\phi$ is $1$. Thus, the satisfaction value of $\varphi$ is $0$ in the worst case, and this is also the highest satisfaction value that $\T_1$ achieves with probability $1$. On the other hand, the expected satisfaction value of $\varphi$ in a computation of $\T_1$ is $p \cdot 0 + (1-p) \cdot 1 = 1-p$. 
Thus, $\sem{\T_1,\phi}_w = \sem{\T_1,\phi}_a = 0$, whereas $\sem{\T_1,\phi}_s = 1-p$.

Consider now a transducer $\T_2$ that closes the connection only on the second cycle.
With probability $p$, we have that $\Next data$ holds, in which case the satisfaction value of $\phi$ is $1$. Also, with probability $1-p$, we have that $\Next data$ does not hold,  in which case the satisfaction value of $\phi$ is $\frac12$. Thus, now the satisfaction value of $\varphi$ is $\frac12$ in the worst case, and this is also the highest satisfaction value that $\T_2$ achieves with probability $1$. On the other hand, the expected satisfaction value of $\varphi$ in a computation of $\T_2$ is $p \cdot 1+(1-p)\cdot \frac12= \frac12 (1+p)$. 
Thus, $\sem{\T_2,\phi}_a =\sem{\T_2,\phi}_w =\frac12$, whereas $\sem{\T_2,\phi}_s = \frac12 (1+p)$. 

To conclude, when $p \geq \frac{1}{3}$, in which case $\frac12 (1+p) \geq 1-p$, then $\T_2$ is superior to $\T_1$ in all the three approaches. When, however, $p < \frac{1}{3}$, then a designer that cares for the expected satisfaction value should prefer $\T_1$. 
\hfill \qed
}
\end{example}

\subsection{The Achievability MDP of an \FLTL formula}
\label{sec:Expected synthesis}

In this section we develop the technical tool we are going to use for solving the high-quality synthesis problem in the stochastic approach. 

Consider an \LTLF formula $\phi$. Let $V(\phi)=\{v_1,...,v_n\}$, with $v_1<...<v_n\in [0,1]$. By
 Theorem~\ref{from abk}, we have that $n \leq 2^{|\phi|}$. Also, for every $1 \leq i \leq n$, there is a \DPW $\A_i$ such that $L(\A_i)=\set{w:\sem{w,\phi}=v_i}$.
Let  $\A_i=\zug{\tAP,Q^i,q_0^i,\delta^i,\alpha^i}$. We construct the product pre-automaton $\A=\A_1\times\ldots\times \A_n$ that subsumes the joint behavior of the \DPWs. Formally, $\A=\zug{\tAP,S,s_0,\mu}$, where $S=Q^1\times\ldots\times Q^n$, the initial state is $s_0=\zug{q^1_0,...,q^n_0}$, and for every state $s=\zug{q_1,...,q_n}$ and $\sigma\in \tAP$, we have $\mu(s,\sigma)=\zug{\delta^1(q_1,\sigma),...,\delta^n(q_n,\sigma)}$.

Every pre-automaton $\B=\zug{\tAP,Q,q_0,\delta}$ 
induces a pre-MDP $\M_\B=\zug{Q,q_0,2^O,\MDPProb}$ in which for every two states $q,q' \in S$ and action $o\in \tOUT$, we have $\MDPProb(q,o,q')=\frac{|\{i \in 2^I : \delta(q,i \cup o)=q'\}|}{2^{|I|}}$. That is, choosing action $o \in 2^O$ in state $q$, the MDP samples the possible inputs $i \in 2^I$ uniformly and moves to state $\delta(q,i \cup o)$.  Consider a memoryless strategy $f:Q \rightarrow 2^O$ for $\M_\B$. The strategy $f$ induces an $I/O$-transducer $\T[\M_\B,f]=\zug{I,O,Q,q_0,\delta',\mu}$ in which for every state $q \in Q$, we have $\mu(q)=f(q)$, and for all $i \in 2^I$, we have $\delta'(q,i)=\delta(i \cup \mu(q))$. Thus, the transducer has the same state space as $\B$, it lets $f$ fix the labels of the states, and uses this label to complete the $2^I$ component of the alphabet to a letter in $\tAP$. 

Consider a parity acceptance condition $\alpha$ on the state space $Q$ of $\B$. Using the notations of \cite{CHJ04a}, a state $q \in Q$ in $\M_\B$ is {\em controllably win recurrent\/} (c.w.r., for short) if there exists an end component $U\subseteq Q$ such that $q\in U$,
$\alpha(q)=\max_{p\in U}\set{\alpha(p)}$, and $\alpha(q)$ is even. That is, $q$ has the maximal rank in $U$, and this rank is even. 
The end component $U$ is referred to as a {\em witness} for $q$ being c.w.r. Intuitively, a parity-MDP with a parity objective  $\alpha$ has a strategy to win with probability $1$ from all c.w.r. states. Moreover, if $U$ is a witness for some c.w.r. state, then there exists a strategy to win with probability $1$ from every state in $U$. 
If, however, a run of $\M_\B$ reaches an end component that does not have a c.w.r. state, then it is winning with probability $0$. 

Once we have defined the product pre-automaton $\A$, we construct an MDP $\M_\A=\zug{S,s_0,2^O,\MDPProb,\MDPre}$, with the following reward function. For a state $s=\zug{q_1,...,q_n}$ of $\M_\A$, we say that a value $v_i \in V(\phi)$ is {\em achievable\/} from $s$ if there exists a c.w.r. state in $\M_{\A_i}$ with a witness $U_i$ for which
$q_i \in U_i$.
Then, $\MDPre(s)=\max\{v_i: v_i$ is achievable from $s \}$. 
Note that the way we have defined $\A$ guarantees that every state that is a part of some end component has at least one value $v_i$ that is achievable from $s$. For states that are not in end components, we define the reward to be $0$.
Intuitively, $\MDPre(s)$ is the highest satisfaction value that can be guaranteed with probability $1$ from $s$.
We refer to $\M_\A$ as the {\em achievability MDP for $\phi$}.

This completes the construction of $\M_\A$. Note that every end component $U$ consists of states with the same value $v_U$. Thus, every infinite run $r$ of $\M$ eventually gets trapped in some end component $U$, implying that $\MDPre(r)=v_U$. Indeed, the rewards along the states in the finite prefix of $r$ that leads to $U$ are averaged out.
For an end-component $U$ of $\M_\A$, let $U|_i$ be the projection of $U$ on $Q^i$. Note that $U|_i$ is an end component in $\A_i$.

\section{Synthesis Against a Stochastic Environment}
\label{sec:Expected synthesis problem}
In the {\em stochastic high-quality synthesis\/} problem (\synprob, for short), we get as input an $\LTLF$ formula $\phi$ over sets $I$ and $O$ of input and output signals, and we seek an $I/O$-transducer that maximizes the expected value of a computation (under a uniform distribution of the inputs). Formally, we want to compute $\max_{\T}\set{\EE[\sem{\T,\phi}_s]}$ and return the witness transducer.\footnote{A-priori, it is not clear that the maximum is attained. As we prove, however, this is in fact the case.}

We solve the \synprob problem by reasoning on the achievable MDP $\M_\A$. Consider a strategy $f$ for $\M_\A$.
Let $\T$ be the transducer induced from $f$, that is $\T=\T[M_\A,f]$. Recall the random variable $X_{\T,\phi}: (2^I)^\omega \rightarrow \RR$ that maps $w \in (2^I)^\omega$ to $\sem{\T(w),\phi}$. We define the random variables $Y_{\T,\phi}:\tINo\to \RR$ as follows. For every $w\in \tINo$, we let $Y_{\T,\phi}(w)$ be the mean-payoff of the values along the run of $\A$ on $\T(w)$. Formally, let $r$ be the run of $\A$ on $\T(w)$. Then, $Y_{\T,\phi}(w)=\MDPre(r)$, where $\MDPre$ is the reward function of $\M_\A$. By definition, we have that 
$\sem{\T,\phi}_s=\EE[X_{\T,\phi}]$. Since $\M_\A$ is obtained by assuming a uniform distribution on the inputs, we have that $\EE[Y_{\T,\phi}]=\val_{\M_\A}(f)$.

\begin{theorem}
\label{thm:M to T shqsyn}
Consider an \FLTL formula $\phi$. Let $\M_\A$ be the achievability MDP for $\phi$. For every value $v \in [0,1]$, there exists a strategy $f$ in $\M_\A$ such that $\val_{\M_\A}(f) \geq v$ iff there exists an $I/O$-transducer $\T$ such that  $\sem{\T,\phi}_s \geq v$. 
Moreover, we can find in time polynomial in $\M_\A$ a memoryless strategy $f$ such that
$\sem{\T[\M_\A,f],\phi}_s$ maximizes $\set{\EE[\sem{\T,\phi}_s]}$.
\end{theorem}
\begin{proof}
We start by proving that if there exists a transducer $\T$ such that $\sem{\T,\phi}_s\ge v$, then there exists a strategy $f$ such that $\val_{\M_\A}(f)\ge v$. For this, we prove, in Appendix~\ref{app thm:M to T shqsyn}, that $\EE[X_{\T,\phi}]\le \EE[Y_{\T,\phi}]$. This is indeed sufficient, as we can then take $f$ to be the strategy induced by $\T$. 

For the converse implication, consider a strategy $f$ in $\M_\A$ such that $\val_{\M_\A}(f)\ge v$. 
By \cite{FV96}, we can assume that $f$ is memoryless. Let $\T=\T[\M_\A,f]$ be the transducer induced by $f$. In Appendix~\ref{app thm:M to T shqsyn}, we show that there exists a transducer $\T'$ such that $\EE[X_{\T',\phi}]= \EE[Y_{\T,\phi}]$, thus concluding the claim.
\end{proof}

We now proceed to show how to solve the \synprob problem. 
\begin{theorem}
\label{thm:solving expected synthesis}
Solving the \synprob problem for \FLTL can be done in doubly-exponential time. The corresponding decision problem is 2EXPTIME complete.
\end{theorem}
\begin{proof}
Consider an \FLTL formula $\phi$. We want to find a transducer $\T$ that maximizes $\sem{\T,\phi}_s$. 
Let $\M_\A$ be the achievability MDP for $\phi$. By Theorem~\ref{thm:M to T shqsyn}, 
we can find in time polynomial in $\M_\A$ a memoryless strategy $f$ such that
$\sem{\T[\M_\A,f],\phi}_s$ maximizes $\set{\EE[\sem{\T,\phi}_s]}$.
Below we analyze the size of $\M_\A$. Let $|\phi|=k$. By Theorem~\ref{from abk}, we have that $n \leq 2^{k}$ and each $\A_i$ is of size at most $2^{2^{O(k)}}$. 
Thus, the size of $\M_\A$ is at most $(2^{2^{O(k)}})^{2^k}=2^{2^{O(k)}}$, implying the doubly exponential upper bound.

A matching lower bound for the respective decision problem follows from the 2EXPTIME hardness of standard LTL synthesis. Note that in our setting one considers satisfaction with probability $1$. Still, since the hardness proof for LTL synthesis considers the interaction between a system and its environment along a finite prefix of a computation (one that models the computation of a Turing machine that halts), it applies also for the stochastic setting. 
\hfill 
\end{proof}

\section{Adding an Almost-Sure Threshold}
\label{sec:Expected synthesis with threshold}
In this section we combine the stochastic and the almost-sure approaches. The {\em \synprob problem with a threshold} includes both an \FLTL formula $\phi$ and a threshold $t\in [0,1]$. The goal then is to maximize the expected satisfaction value of $\phi$ while guaranteeing that it is almost surely above $t$.
Formally, given $\phi$ and $t$, we seek a transducer $\T$ that maximizes $$\set{\sem{\T,\phi}_s: \sem{\T,\phi}_a \geq t}.$$
Note that there need not be a transducer $\T$ for which $\sem{\T,\phi}_a \geq t$, in which case the set is empty and we return no transducer. This is the multi-valued analogue of an unrealizable Boolean specification (except that here the user may want to try to reduce $t$). Note also that this sub-problem, of deciding whether the set is empty, amounts to solving the high-quality synthesis problem in the almost-sure approach. Finally, if the set is not empty, then we have to show, as in Section~\ref{sec:Expected synthesis problem}, that its maximum is indeed attained.

\begin{example}
\label{xpl:message sending}
{\rm 
Consider a server sending messages over a noisy channel. At each cycle, the server sends a message and 
needs to decide whether to encode it so that error-correction can retrieve it in case the channel is noisy,  or take a risk and send the message with no encoding.
Encoding a message has some cost. 
We
formulate the quality of each cycle by the specification $\psi$ over $I=\{{\it noise}\}$ and $O=\{{\it encode}\}$, where
$\psi=(\neg {\it noise} \wedge \neg {\it encode})\vee \factorU_{\frac34}{\it encode}.$
Thus, each cycle has satisfaction value $1$ if a message that is not encoded is sent over a non-noisy channel, and satisfaction value $\frac34$ if a message is encoded. Note that otherwise (that is, when a message that is not encoded is sent over a noisy channel), the satisfaction value is $0$. The factor $\frac34$ in the \FLTL specification reflects the priorities of the designer as induced by the actual cost of encoding and of losing messages.

Recall that $\psi$ specifies the quality of a single cycle. The quality of a full computation refers to its different cycles, and a natural thing to do is to take the average over the cycles we want to consider. Assume that a channel may be noisy only during the first four cycles. Then, the quality of a computation is  
$\phi=(\psi\avg{\frac12} \Next\psi) \avg{\frac12}(\Next\Next\psi\avg{\frac12}\Next\Next\Next\psi).$

Assume that the probability of a channel to be noisy in each of the first four cycles is $p$.
Consider a transducer $\T_1$ that does not encode any message. The expected satisfaction value of $\psi$ in each of the four cycles is then $(1-p)\cdot 1+p\cdot 0 = 1-p$, hence $\sem{\T_1,\phi}_s=1-p$. On the other hand, the satisfaction value of $\psi$ in a noisy cycle is $0$, hence $\sem{\T_1,\phi}_w=\sem{\T_1,\phi}_a=0$. 
Thus, if one does not care for a lower bound on the satisfaction value in the worst case, then by using $\T_1$ he gets an expected satisfaction value of $1-p$. 

Suppose now that we want the satisfaction value to be above $\frac13$ in the worst case. 
This can be achieved by a transducer $\T_2$ that encodes messages in two of the four cycles. Indeed, for the cycles in which a message is encoded, we get satisfaction value $\frac34$, which is averaged with $0$, namely the worst-case satisfaction value in the cycles in which a message is not encoded. Hence, $\sem{\T_2,\phi}_w=\sem{\T_2,\phi}_a=\frac34\avg{\frac12} 0=\frac38 > \frac13$. The expected satisfaction value of $\T_2$ is then $\sem{\T_2,\phi}_s = \frac34 \avg{\frac12} (1-p) =\frac{7}{8}-\frac{p}{2}$. 

Finally, if we want to ensure satisfaction value $\frac34$ in the worst case, then we can design a transducer $\T_3$  that encodes all the messages in the first four cycles. Now, $\sem{\T_3,\phi}_w=\sem{\T_3,\phi}_a=\sem{\T_3,\phi}_s=\frac34$. 

It follows that for a small $p$, adding a threshold on the satisfaction value in the worst case reduces the expected satisfaction value. Indeed, when $p <\frac14$, then $1-p > \frac{7}{8}-\frac{p}{2} > \frac34$. When, however, $p \geq \frac14$, then $\T_3$ is superior in the three approaches.
\hfill \qed
}
\end{example}

In order to solve the \synprob problem with a threshold, we modify our solution from Section~\ref{sec:Expected synthesis} as follows.
We start by deciding whether there exists a transducer $\T$ such that $\sem{\T,\phi}_a\ge t$. 
For this, we construct, per Theorem~\ref{from abk}, a DPW
$\A_{\ge t}=\zug{\tAP,Q^{\ge t},q_0^{\ge t},\delta^{\ge t},\alpha^{\ge t}}$ such that $L(\A_{\ge t})=\set{w: \sem{w,\phi}\ge t}$. Let $\M_{\ge t}$ be the parity-MDP induced from $\A_{\ge t}$. By~\cite{CHJ04a}, we can find the set of almost-sure winning states of $\M_{\ge t}$. If $q_0^{\ge t}$ is winning, then the required transducer exists, and in fact $\M_{\ge t}$ embodies all candidate transducers. We obtain a pre-automaton $\A'_{\ge t}$ from $\A_{\ge t}$ by restricting $\A_{\ge t}$ to winning states, and removing transitions from state $q\in Q^{\ge t}$ for every action $o\in \tOUT$ such that there exists $i\in \tIN$ for which $\delta^{\ge t}(q,i\cup o)$ is not a winning state.  

We proceed by constructing a product pre-automaton $\A$ that is similar to the one constructed in Section~\ref{sec:Expected synthesis}, except that takes $t$ and $\A'_{\ge t}$ into account, as follows.

Let $\ell=\arg\min_i\set{v_i: v_i\ge t}$ be the minimal index such that $v_i \ge t$. We define $\A=\A_\ell\times \ldots \times  \A_n\times 
\A'_{\ge t}$. That is, the product, defined as in Section~\ref{sec:Expected synthesis}, now contains only DPWs $\A_i$ for which $v_i \ge t$ and also contains $\A'_{\ge t}$. We obtain the MDP $\M_\A$ and set the reward function as in Section~\ref{sec:Expected synthesis}, taking into account only c.w.r. states from the automata $\A_\ell,...,\A_n$. The component $\A'_{\ge t}$ is only used to restrict the actions of the MDP $\M_\A$. 
We refer to $\M_\A$ as the {\em $t$-achievability MDP for $\phi$}.

We present an analogue to Theorem~\ref{thm:M to T shqsyn}. The proof appears in Appendix~\ref{apx:M to T shqsyn with threshold}.
\begin{theorem}
\label{thm:M to T shqsyn with threshold}
Consider an \FLTL formula $\phi$ and a threshold $t \in [0,1]$. Let $\M_\A$ be the $t$-achievability MDP for $\phi$. For every value $v \in [0,1]$, there exists a strategy $f$ in $\M_\A$ such that $\val_{\M_\A}(f) \geq v$ iff there exists an $I/O$-transducer $\T$ such that $\sem{\T,\phi}_a\ge t$ and $\sem{\T,\phi}_s \geq v$. 
Moreover, we can find in time polynomial in $\M_\A$ a memoryless strategy $f$ such that
$\sem{\T[\M_\A,f],\phi}_s$ maximizes $\set{\EE[\sem{\T,\phi}_s] : \sem{\T,\phi}_a \ge t}$.
\end{theorem}

Since, by Theorem~\ref{from abk}, the size of $\A_{\ge t}$ is doubly exponential in $\phi$, then, by following considerations  similar to these specified in the proof of Theorem~\ref{thm:solving expected synthesis}, we conclude with the following.
\begin{theorem}
\label{thm:solving expected synthesis with threshold}
Solving the \synprob problem with a threshold for \FLTL can be done in doubly-exponential time. The corresponding decision problem is 2EXPTIME-complete.
\end{theorem}

\begin{remark}
\label{rmk:threshold using existing tools}
In~\cite{CKK15,CR15}, the authors solve the problem of deciding, given an MDP $\M$ and two thresholds $v$ and $t$, whether there is a strategy $f$ for $\M$ that guarantees value $t$ almost surely, and has expected cost at least $v$. The solution can be directly applied to our setting. However, note that this solution only guarantees an expected cost of $v$, whereas our approach finds the optimal expected cost.
\end{remark}
\begin{remark}
In the \synprob problem with a threshold, we use the formula $\phi$ both for the expectation maximization, and for the almost-sure threshold. Sometimes, it is desirable to decompose the specification into one part -- $\psi$, which is a hard constraints and needs to be satisfied almost-surely above the threshold $t$, and another part -- $\phi$, which specifies a utility function with respect to which we would like to optimize \cite{BFRR14a,CR15}.   

Our solution can be easily adapted to handle this setting. Indeed, in the construction of the $t$-achievability MDP, we replace $\A_{\ge t}$, with $\B_{\ge t}$, where $L(\B_{\ge t})=\set{w: \sem{w,\psi}\ge t}$, and proceed with the described construction and the proofs.
\end{remark}

\section{Adding Environment Assumptions}
\label{sec:assume-guarantee}
A common paradigm in Boolean synthesis is synthesis with environment assumptions~\cite{CHJ08,LDS11}, where the input to the synthesis problem consists of a specification $\phi$ and an assumption $\psi$, and we seek a transducer that realizes $\phi$ under the assumption that the environment satisfies $\psi$. 
In this section we consider an analogue variant of the \synprob problem, where we are given an $\FLTL$ specification $\phi$ and an $\LTL$ assumption $\psi$, and we seek a transducer that maximizes the expected satisfaction value of $\phi$ given that the environment satisfies the assumption $\psi$. Note that while the specification is quantitative, the assumption is Boolean. 

\begin{example}
\label{xpl:message sending assumption}
{\rm 
Recall the message-sending server in Example~\ref{xpl:message sending}, and assume that the channel can change its status (noisy/non-noisy) only every second cycle. We use this assumption in order to design improved transducers. 
Formally, the assumption is given by the $\LTL$ formula $\psi=({\it noise}\leftrightarrow \Next{\it noise})\wedge \Next\Next({\it noise}\leftrightarrow \Next{\it noise})$.

The transducer $\T_4$ does not encode the first message, but checks whether the channel was noisy. If it was, the second message is encoded. We get that the expected satisfaction value of $\phi$ in $\T_4$ under the assumption is $(1-p+p\cdot\frac34+(1-p)\cdot 1)/2=1-\frac{5}{8}p$, which is higher than $1-p=\sem{\T_1,\phi}_s$ for every $p>0$. In addition, under the assumption we are guaranteed that the worst-case satisfaction value of $\T_4$ is at least $\frac38$, unlike $\T_1$ (in case the channel is noisy, so only the second and fourth messages are encoded). 
Thus $\T_4$ is superior to $\T_1$ described in Example~\ref{xpl:message sending} in the three approaches (under the assumption).

Next, as in Example~\ref{xpl:message sending}, if we want to ensure satisfaction value $\frac34$ in the worst case, we can design a transducer $\T_5$ that works like $\T_4$, except that it always encodes the first and third messages.
The expected satisfaction value of $\T_5$ under the assumption is
$(\frac34 + p\cdot \frac34 + (1-p)\cdot 1)/2=\frac78-\frac{p}{8}$, which is higher than $\frac34=\sem{\T_3,\phi}_s$, for every $p\in[0,1]$.

Thus, under the assumption, it is possible to design transducers that increase the expected satisfaction value as well as the lower bound. 
\hfill \qed
}
\end{example}

Formally, in the {\synprob problem with environment assumptions}, we get as input an \FLTL formula $\phi$ over $I\cup O$, and an {\em environment assumption} $\psi$, which is an $\LTL$ formula over $I$ such that $\Pr(\psi) > 0$. That is, the probability of the event $\{w : w \models \psi\} \subseteq (2^I)^\omega$ is strictly positive. Recall that $X_{\T,\phi}$ is a random variable such that $X_{\T,\phi}(w)=\sem{\T(w),\phi}$. We seek a transducer $\T$ that maximizes $\EE[X_{\T,\phi}|w\models \psi]$.

We start by citing a folklore lemma,
whose proof can be found in Appendix~\ref{apx:expectation}.
\begin{lemma}
\label{lem:condition on union with 0}
Consider a random variable $X$. Let $A,B$ be events such that $\Pr(A)>0$ and $\Pr(B)=0$. Then, $\EE[X|A\cup B]=\EE[X|A]$.
\end{lemma}
Before proceeding, we note that if $\Pr(\psi)=1$, then we can proceed by dropping the assumption entirely. Indeed, it holds that 
$\Pr(\neg\psi)=0$, and by Lemma~\ref{lem:condition on union with 0}, we have that $\EE[X_{\T,\phi}|w\models \psi]=\EE[X_{\T,\phi}|(w\models \psi)\cup (w\models\neg \psi)]=\EE[X_{\T,\phi}|(2^I)^\omega]=\EE[X_{\T,\phi}]$. Thus, we henceforth assume that $0<\Pr(\psi)<1$.

As mentioned in Section~\ref{sec:intro}, maximizing the conditional expectation directly is notoriously problematic, as, unlike unconditional expectation, it is not a linear objective. Thus, it is not susceptible to linear optimization techniques, which are the standard approach to find maximizing strategies in MDPs. 
Our solution is a modification of the construction from Section~\ref{sec:Expected synthesis} in which we, intuitively, ``redistribute'' the probability of the input sequences that do not satisfy the assumption. 
We start by constructing a DPW $\A_\psi$ that accepts a word $w\in (2^I)^\omega$ iff $w\models \psi$. Note that the alphabet of $\A$ is $\tIN$. We think of this alphabet as $\tAP$, where transitions simply ignore the $\tOUT$ component. In particular, the MDP $\M_{\A_\psi}$ is in fact an MC. We say that an ergodic component of $\M_{\A_\psi}$ is {\em rejecting} if the maximal rank that appears in it is odd. It is easy to see that a run in a rejecting ergodic component is accepting w.p. 0.

We then consider the automaton $\A=\A_\psi\times \A_1\times\ldots\times \A_n$, and obtain the MDP 
$\M_\A=\zug{S,s_0,2^O,\MDPProb,\MDPre}$ as described in Section~\ref{sec:Expected synthesis}. In particular, the reward function is as there, and the only change is the addition of the $\A_\psi$ component, which provides information about satisfaction of $\psi$. We refer to $\M_\A$ as the {\em conditional achievability MDP for $\phi$ given $\psi$}.
Recall that for a strategy $f$, we have defined $R_{\M_A,f}$ as a random variable whose value is the reward on runs in ${\M_\A}$ with strategy $f$. 
Following the proof of Theorem~\ref{thm:M to T shqsyn}, we then get the following.

\begin{theorem}
\label{thm:M to T shqsyn conditional}
Consider an \FLTL formula $\phi$ and an environment assumption $\psi$. Let $\M_\A$ be the conditional achievability MDP for $\phi$ given $\psi$. For every value $v \in [0,1]$, there exists a strategy $f$ in $\M_\A$ such that
$\EE[R_{\M_\A,f}|w\models \psi]\ge v$
iff there exists an $I/O$-transducer $\T$ such that $\EE[X_{\T,\phi}|w\models \psi]\ge v$. 
Moreover, if $f$ is memoryless, then we can find in time polynomial in $\M_\A$ a memoryless strategy $f$ such that
$\EE[X_{\T[\M_\A,f],\phi}|w\models \psi]\ge v$.
\end{theorem}
Theorem~\ref{thm:M to T shqsyn conditional} enables us to reason about $\M_\A$, but we are still left with conditional expectations. To handle the latter, 
we follow a technique suggested in \cite{BKKM14} and obtain from $\M_\A$ a new MDP ${\M'_\A}=\zug{S,s_0,A,\MDPProb',\MDPre}$ as follows. A state $s=\zug{q,q_1,...,q_n}$ of $\M_{\A}$ is called a {\em rejecting ergodic state} if its state $q$ of $\A_\psi$ belongs to a rejecting ergodic component of $\M_{\A_\psi}$. 
Let $\baderg=\{s:s\text{ is a rejecting ergodic state}\}$. 

For every state $s\in \baderg$ we set $\MDPProb'(s,a,s_0)=1$. That is, whenever a rejecting ergodic component of $\A_\psi$ is reached, the MDP ${\M'_\A}$ deterministically resets back to $s_0$.

Intuitively, when a rejecting ergodic component of $\A_\psi$ is reached, then the probability of $\psi$ being satisfied is $0$. Thus, resetting ``redistributes'' the probability of $\psi$ not being satisfied evenly. 
Below we formalize this intuition. The proofs can be found in Appendices~\ref{apx:M' to M} and~\ref{apx:M to M'}.
\begin{lemma}
\label{lem:M' to M}
Let $v\in \RR$, and consider a memoryless strategy $g$ in ${\M'_\A}$ such that $\val_{{\M'_\A}}(g)\ge v$. There exists a memoryless strategy $f$ in $\M_\A$ such that $\EE[R_{\M_\A,f}|w\models\psi]\ge v$. 
Moreover, $f$ can be computed from $g$ in polynomial time.
\end{lemma}

\begin{lemma}
\label{lem:M to M'}
Let $v\in \RR$, and consider a strategy $f$ in $\M_\A$ such that $\EE[R_{\M_\A,f}|w\models\psi]\ge v$.  There exists a strategy $g$ in ${\M'_\A}$ such that $\EE[R_{{\M'_\A},g}]\ge v$.
\end{lemma}
%
 
Finally, using Theorem~\ref{thm:M to T shqsyn conditional}, and the fact that $\A_\psi$ is doubly exponential in $\psi$, we can use the same reasoning as in the proof of Theorem~\ref{thm:solving expected synthesis} and conclude with the following. The proof can be found in Appendix~\ref{apx:solving expected synthesis conditional}.
\begin{theorem}
\label{thm:solving expected synthesis conditional}
Solving the \synprob problem with environment assumptions 
can be done in doubly-exponential time. The corresponding decision problem is 2EXPTIME-complete.
\end{theorem}
%

\section{Extensions}
\label{sec:extensions}
In this section we describe two extensions to the setting. The first combines the threshold and assumption extensions presented in Sections~\ref{sec:Expected synthesis with threshold} and~\ref{sec:assume-guarantee}. The second shows how to handle a non-uniform probability distribution.
\subsection{Combining an Almost-Sure Threshold with Environment Assumptions}
Combining an almost-sure threshold with environment assumptions requires some subtlety in the definitions. As an input for the problem, we are given an \LTLF formula $\phi$ over $I\cup O$, an \LTL environment assumption $\psi$ over $I$ such that $\Pr(\psi)>0$, and a threshold $t\in [0,1]$. Then, we seek a a transducer $\T$ that maximizes $\EE[X_{\T,\phi}|w\models \psi]$ and for which $\Pr(\sem{\T(w),\phi}\ge t|w\models \psi)=1$. In particular, the threshold $t$ should be attained almost surely only in computations that satisfy $\psi$. 

\begin{remark}
{\rm
Note that it could have also been possible to seek a transducer $\T$ that maximizes $\EE[X_{\T,\phi}|w\models \psi]$ and for which $\Pr(\sem{\T(w),\phi}\ge t)=1$, namely for which the threshold should hold almost surely regardless of the assumption. We found this approach less appealing. Its solution, however, is a straightforward combination of our constructions. That is, we start with the product $\A_{\ell}\times\ldots \times \A_n\times \A'_{\ge t}\times \A_\psi$, as defined in Sections~\ref{sec:Expected synthesis with threshold} and \ref{sec:assume-guarantee}, apply the reset modification described in Section~\ref{sec:assume-guarantee}, and seek a maximizing strategy in the resulting MDP. 
\hfill \qed}
\end{remark}
We solve the problem as follows. We start by checking whether there exists a transducer $\T$ such that $\Pr(\sem{\T(w),\phi}\ge v|w\models \psi)=1$, using the following lemma (see Appendix~\ref{apx:conditional threshold} for the proof). 
\begin{lemma}
\label{lem: conditional threshold}
Let $\phi,\psi$, and $t$ be as above. For every transducer $\T$ it holds that $\Pr(\sem{\T(w),\phi}\ge t|w\models \psi)=1$ iff $\Pr(\sem{\T(w), \psi \to \phi}\ge t)=1$.
\end{lemma}
Using Lemma~\ref{lem: conditional threshold}, we can decide the existence of a transducer $\T$ as we seek, by constructing the \DPW $\A_{{\psi\to \phi},{\ge t}}$ as per Theorem~\ref{from abk}, and keeping only almost-sure winning states as done in Section~\ref{sec:Expected synthesis with threshold}.

We now proceed as in the first approach, by constructing the product $\A_{\ell}\times\ldots \times \A_n\times \A'_{{\psi\to \phi},{\ge t}}\times \A_\psi$, where $\A'_{{\psi\to \phi},{\ge t}}$ is obtained from $\A_{{\psi\to \phi},{\ge t}}$ by keeping only almost-sure winning states. 

\subsection{Handling a Non-Uniform Distribution}
\label{nud}
In order to handle a non-uniform distribution on the input signals, we first have to decide how to model arbitrary distributions on $\tINo$. The common way to do so is to assume that the distribution is generated by a pre-MDP ${\cal D}=\zug{S,s_0,\tOUT,\MDPProb}$ and a labeling function $\iota:S\to \tIN$, where a state $s\in S$ generates the input $\iota(s)$. Thus, the probability of an input signal to hold depends on the history of the interaction with the system. Formally, every run $r=s_0,s_1,...$ of ${\cal D}$ generates an input sequence $\iota(s_1),\iota(s_2)$, and the distribution on runs induces a distribution on $\tINo$. \footnote{Note that we do not consider the label on $s_0$, in order to allow a distribution on the initial letters.}

All our results can be adapted to handle a distribution given by ${\cal D}$ as above. We only have to change the construction of the achievability MDP described in Section~\ref{sec:Expected synthesis} as follows. For a pre-automaton $\B=\zug{\tAP,Q,q_0,\delta}$ and a distribution pre-MDP ${\cal D}=\zug{S,s_0,\tOUT,\MDPProb}$ with labeling function $\iota$, we define the induced pre-MDP as $\M_\B^{\cal D}=\zug{Q\times S,\zug{q_0,s_0},\tOUT,\MDPProb'}$ where for every two states $\zug{q,s},\zug{q',s'} \in Q\times S$ and action $o\in \tOUT$, 
we have 
$\MDPProb'(\zug{q,s},o,\zug{q',s'})=
\MDPProb(s,o,s')$ if $\delta(q,\iota(s')\cup o)=q'$, and $0$ otherwise.
It is not hard to see that all the constructions we apply to achievability MDP $\M_\A$ can be applied to $\M_\A^{\cal D}$, which would take the distribution in ${\cal D}$ into account. The complexity of the algorithms is polynomial in $\M_\A^{\cal D}$. Thus, the complexity of our algorithms remains 2EXPTIME-complete in $\phi$ and polynomial in ${\cal D}$. 

\bibliographystyle{abbrvnat}
\bibliography{ok}

\appendix
\section{Analysis of Example~\ref{xpl:intro hard drive 1}}
\label{apx:example1}
We consider a transducer that replaces the battery in the first station it encounters starting from position $t$, for  $1\le t\le k$. The expected cost of the transducer is then 
\begin{align*}
&(1-p)^k+p\frac{t}{k}+(1-p)p\frac{t+1}{k}+...+(1-p)^{k-t}p\frac{t+(k-t)}{k}\\
&=(1-p)^k+\sum_{i=0}^{k-t}(1-p)^ip\frac{t+i}{k}\\
&=(1-p)^k+\sum_{i=0}^{k-t}(1-p)^ip\frac{t}{k}+\sum_{i=0}^{k-t}(1-p)^ip\frac{i}{k}\\
&=(1-p)^k+p\frac{t}{k}\sum_{i=0}^{k-t}(1-p)^i+\frac{p}{k}\sum_{i=0}^{k-t}(1-p)^ii\\
\end{align*}
\begin{align*}
&=(1-p)^k+p\frac{t}{k}\left(\frac{1-(1-p)^{k-t+1}}{p}\right)+\frac{p(1-p)}{k}\sum_{i=0}^{k-t}(1-p)^{i-1}i\\
&=(1-p)^k+\frac{t}{k}(1-(1-p)^{k-t+1})+\frac{p(1-p)}{k}\left(-\sum_{i=0}^{k-t}(1-p)^{i}\right)'\\
&=(1-p)^k+\frac{t}{k}(1-(1-p)^{k-t+1})+\frac{p(p-1)}{k}\left(\frac{1-(1-p)^{k-t+1}}{p}\right)'\\
\end{align*}
\begin{align*}
&=(1-p)^k+\frac{t}{k}(1-(1-p)^{k-t+1})+\\
+&\frac{p(p-1)}{k}\left(\frac{(k-t+1)(1-p)^{k-t}p-1+(1-p)^{k-t+1}}{p^2}\right)\\
&=(1-p)^k+\frac{t}{k}(1-(1-p)^{k-t+1})+\\
+&\frac{(p-1)}{kp}\left({(k-t+1)(1-p)^{k-t}p-1+(1-p)^{k-t+1}}\right)\\
\end{align*}
One now sees, for example, that if $t=\alpha k$ for $\alpha\in (0,1)$, then the latter expression tends to $\frac{\alpha k}{k}=\alpha$ as $k\to\infty$, as the first and third summands tend to $0$. In particular, for $\alpha=\frac12$, we get an expected satisfaction value of $\frac12$.\qed

\section{Proofs}

\subsection{Full Proof of Theorem~\ref{thm:M to T shqsyn}}
\label{app thm:M to T shqsyn}

We start by proving that if there exists a transducer $\T$ such that $\sem{\T,\phi}_s\ge v$, then there exists a strategy $f$ such that $\val_{\M_\A}(f)\ge v$. To this end, it suffices to prove that $\EE[X_{\T,\phi}]\le \EE[Y_{\T,\phi}]$. Indeed, we can then take $f$ to be the strategy induced by $\T$. 

Consider a random word $w\in \tINo$. it is well-known that for every transducer $\T$, w.p. 1 the run of $\T$ on $w$ reaches an end component and visits all the states of that end component infinitely often (see e.g. \cite{CHJ04a}). Let $U$ be the end component that the run $r$ of $\M_\A$ on $\T(w)$ reaches and for which $\Inf(r)=U$. Let $1\le i\le n$ be such that $v_i=\sem{\T(w),\phi}=X_{\T,\phi}(w)$. Then, $\A_i$ accepts $\T(w)$, and the component $U|_i$ contains a c.w.r. state $q_i$. Indeed, since $\Inf(r)=U$, then the run of $\A_i$ on $w$ visits infinitely often all the states in $U|_i$, implying that the maximal rank in $U|_i$  is even, and that every state that attains this rank in $U|_i$ is a c.w.r. state in $\M_{\A_i}$. Thus, by construction, all the states in $U$ have reward at least $v_i$ (it may be the case that $U$ is contained in another end component with a higher-value c.w.r.). Thus, $Y_{\T,\phi}(w)\ge v_i$. Since our assumption on $\T(w)$ reaching an end component holds w.p. 1, it follows that $\Pr(X_{\T,\phi}\le Y_{\T,\phi})=1$. By taking expectation, we conclude that $\EE[X_{\T,\phi}]\le \val_{\M_\A}(\T)=\EE[Y_{\T,\phi}]$.

For the converse implication, consider a strategy $f$ in $\M_\A$ such that $\val_{\M_\A}(f)\ge v$. 
By \cite{FV96}, we can assume that $f$ is memoryless. Let $\T=\T[\M_\A,f]$ be the transducer induced by $f$. We show that there exists a transducer $\T'$ such that $\EE[X_{\T',\phi}]= \EE[Y_{\T,\phi}]$, thus concluding the claim.

By~\cite{CHJ04a}, if $q$ is a c.w.r. state in an MDP $\M$, and $U$ is a witness for $q$, then there exists a memoryless strategy $g$ such that for every state $q'\in U$, w.p. 1 the run $r$ of $\M$ from $q'$ visits $q$ infinitely often and stays forever in $U$.

We obtain $\T'$ as follows. Once the run of $\M_\A$ with $\T$ reaches an end component $U$, if the states in $U$ have value $v_i$ for some $1\le i\le n$, then $\T$ starts playing the memoryless strategy mentioned above to visit a state $s=\zug{q_1,...,q_n}\in U$ such that $q_i$ is a c.w.r. in $\M_{\A_i}$. Such a state must exist by the construction of $\M$.

Note that the runs of $\T$ and $\T'$ only differ after reaching an end component, in which case while the runs may differ, the values do not differ, as all the states in an end component have the same value. Thus, $\EE[Y_{\T,\phi}]=\EE[Y_{\T',\phi}]$. 

Observe now that once $\T'$ reaches an end component $U$ as above, then w.p. 1 the run visits $q_i$ infinitely often, and is therefore accepting in $\A_i$, implying that $\sem{\T(w),\phi} \geq v_i$. By the construction of $\M_\A$, we have that $v_i$ is the maximal value for which there exists a c.w.r. state in (a projection to the $\A_i$ automata on) $U$. Thus, $\Pr(Y_{\T',\phi}= X_{\T',\phi})=1$. We conclude that $\EE[Y_{\T,\phi}]=\EE[Y_{\T',\phi}]= \EE[X_{\T',\phi}]$, and we are done. 

Finally, it is easy to see that finding the c.w.r states and constructing $\T'$ from $f$ can be done in polynomial time. The first involves finding the winner in parity-MDPs,
and the second follows from the fact that finding the strategies $g$ above can be done in polynomial time~\cite{CHJ04a}. Then, $\T'=\T[\M_\A,f']$, where $f'$ is the strategy that plays $f$ until reaching an end component, and then plays $g$, as described above.

\subsection{Proof of Theorem~\ref{thm:M to T shqsyn with threshold}}
\label{apx:M to T shqsyn with threshold}
For the first direction, namely constructing $f$ given $\T$, the proof is analogous to that of Theorem~\ref{thm:M to T shqsyn}, keeping in mind that $\sem{\T,\phi}_a\ge t$ implies that 
w.p. 1 a run of $\M_\A$ with the strategy $\T$ reaches an end component that contains a c.w.r. state. Indeed, the assumption implies that w.p. 1 the component $\A'_{\ge t}$ accepts $\T(w)$, which means that w.p. 1 at least one of the automata $\A_\ell,...,\A_{n}$ accepts $\T(w)$, so the end component that is eventually reached (also w.p. 1) has a c.w.r. state. The rest of the analysis follows the proof of Theorem~\ref{thm:M to T shqsyn}.

For the other direction, consider a memoryless strategy $f$ in $\M_\A$. Assume that there exists $\epsilon>0$ such that the run of $\M_\A$ with $f$ reaches an end component $U$ that does not have a c.w.r. state w.p. $\epsilon>0$. By the construction of $\M_\A$, all the states in $U$ have reward $0$. Thus, changing the behavior of $\T$ from the states in $U$ cannot decrease the expected value. Furthermore, by the construction of $\A'$, the projection of $U$ on $\A'_{\ge t}$ consists only of states from which there is a strategy that wins w.p. 1 in the parity-MDP $\M_{\ge t}$. Thus, we can modify $f$ to play such a strategy from $U$ while not decreasing the expected value, but guaranteeing that $f$ reaches  w.p. 1 an end component that contains a c.w.r. state. From here, we obtain $\T'$ similarly to the proof of Theorem~\ref{thm:M to T shqsyn}. As there, the strategy can be memoryless and be found in polynomial time. 
\qed

\subsection{Proof of Lemma~\ref{lem:condition on union with 0}}
\label{apx:expectation}
Recall that the probabilistic distribution function of $X$ given an event $C$ with $\Pr(C)>0$ is a function $f_{X|C}$ that satisfies for all $D\subseteq \RR$ 
$$\int_{D}f_{X|C}(x)dx=\Pr(D|C).$$

It is easy to show that since $\Pr(B)=0$, then for every event $D$ it holds that $\Pr(D|A\cup B)=\Pr(D|A)$. Thus, the former condition implies that $f_{X|A\cup B}\equiv f_{X|A}$, and finally
$$\EE[X|A\cup B]=\int_\RR xf_{X|A\cup B}=\int_\RR xf_{X|A}=\EE[X|A].$$
\qed

\subsection{Proof of Lemma~\ref{lem:M' to M}}
\label{apx:M' to M}
We construct $f$ to agree with $g$ on every state not in $\baderg$. On states in $\baderg$, we set $f$ to behave arbitrarily. We claim that $\EE[R_{\M_\A,f}|w\models\psi]=\val_{{\M'_\A}}(g)$. 

Let $B \subseteq  \tINo$ be the event such that $w \in B$ iff the run of $\A_\psi$ on $w$ reaches a rejecting ergodic component. Equivalently, this is set of words for which a run of $\M_\A$ reaches $\baderg$.  
Recall that $\val_{{\M'_\A}}(g)=\EE[R_{{\M'_\A},g}]$. 
By the law of total expectation we get
$\EE[R_{{\M'_\A},g}]=\EE[R_{{\M'_\A},g}|B]\cdot\Pr(B)+\EE[R_{{\M'_\A},g}|\overline{B}]\cdot\Pr(\overline{B}).$
Since $g$ is memoryless, and since visiting $B$ in ${\M'_\A}$ implies a reset to $s_0$, we get that $\EE[R_{{\M'_\A},g}|B]=\EE[R_{{\M'_\A},g}]$. Thus, the above becomes
$\EE[R_{{\M'_\A},g}]=\EE[R_{{\M'_\A},g}]\cdot\Pr(B)+\EE[R_{{\M'_\A},g}|\overline{B}]\cdot(1-\Pr(\overline{B})).$
Rearranging and dividing by $(1-\Pr(B))$, which is nonzero since we assume $\Pr(\psi)<1$, we get that $\EE[R_{{\M'_\A},g}]=\EE[R_{{\M'_\A},g}|\overline{B}]$.

Next, we observe that given $\overline{B}$, the behavior of $f$ and $g$ is identical, since the reset states are never reached. Thus, we get $\EE[R_{{\M'_\A},g}]=\EE[R_{{\M'_\A},g}|\overline{B}]=\EE[R_{\M_\A,f}|\overline{B}]$.

We now partition the event $\{w: w\models \psi\}$ to
$(\set{w: w\models \psi}\cap B)\cup (\set{w: w\models \psi}\cap \overline{B})$. Observe that, by definition, $\Pr(\set{w: w\models \psi}\cap B)=0$. Therefore, by Lemma~\ref{lem:condition on union with 0}, we get that $\EE[R_{\M_\A,f}|w\models\psi]=\EE[R_{\M_\A,f}|w\models\psi\cap \overline{B}]$. Similarly, we partition the event $\overline{B}$ to $\overline{B}=(\set{w: w\models \psi}\cap \overline{B}) \cup (\set{w: w\not\models \psi}\cap \overline{B})$. Observe that $\Pr(\set{w: w\not\models \psi}\cap \overline{B})=0$. Indeed, given that a computation does not reach $\baderg$, it reaches w.p. 1 an ergodic state from which $\psi$ is satisfied w.p. 1, and therefore the computation satisfies $\psi$ w.p. 1. Again, using Lemma~\ref{lem:condition on union with 0} and the above observation, we get
$\EE[R_{\M_\A,f}|\overline{B}]=\EE[R_{\M_\A,f}|w\models\psi\cap \overline{B}]=\EE[R_{\M_\A,f}|w\models\psi]$,
implying that $\EE[R_{{\M'_\A},g}]= \EE[R_{\M_\A,f}|\overline{B}] = \EE[R_{\M_\A,f}|w\models\psi]$.
\qed

\subsection{Proof of Lemma~\ref{lem:M to M'}}
\label{apx:M to M'}
We construct $g$ to behave as follows. As long as $\baderg$ is not reached, $g$ behaves as $f$. By the definition of ${\M'_\A}$, once a state in $\baderg$ is reached, the next step resets to $s_0$. We then let $g$ ``start over'' and again behave as $f$ does on an empty history.
We claim that $\EE[R_{{\M'_\A},g}]=\EE[R_{\M_\A,f}|w\models\psi]$. The proof is similar to that of Lemma~\ref{lem:M' to M}.

Consider the event $B$ as in the proof of Lemma~\ref{lem:M' to M}.
Note that since $g$ resets whenever $\baderg$ is reached, we have that $\EE[R_{{\M'_\A},g}]=\EE[R_{{\M'_\A},g}|B]$. Since
$\EE[R_{{\M'_\A},g}]=\EE[R_{{\M'_\A},g}|B]\cdot\Pr(B)+ \EE[R_{{\M'_\A},g}|\overline{B}]\cdot\Pr(\overline{B}),$
we rearrange and get $\EE[R_{{\M'_\A},g}]= \EE[R_{{\M'_\A},g}|\overline{B}]$.
Again, as in the proof of Lemma~\ref{lem:M' to M}, we have that $\EE[R_{\M_\A,f}|w\models \psi]=\EE[R_{\M_\A,f}|\overline{B}]$.
Finally, since $f$ and $g$ coincide given $\overline{B}$, we conclude that
$\EE[R_{{\M'_\A},g}]=\EE[R_{{\M'_\A},g}|B]=\EE[R_{\M_\A,f}|\overline{B}]=\EE[R_{\M_\A,f}|w\models \psi]$ and we are done.
\qed

\subsection{Proof of Theorem~\ref{thm:solving expected synthesis conditional}}
\label{apx:solving expected synthesis conditional}
Consider an $\FLTL$ specification $\phi$ over $I\cup O$ and an $\LTL$ assumption $\psi$ over $I$. 
By Theorem~\ref{thm:M to T shqsyn conditional}, it is enough to find a memoryless strategy $f$ in $\M_\A$ that maximizes $\EE[R_{\M_\A,f}|w\models \psi]$. Consider a memoryless strategy $g$ that maximizes $\val_{\M'_\A}(g)$. By Lemma~\ref{lem:M' to M}, we can compute in polynomial time in $\M'_\A$ a memoryless strategy $f$ such that $\EE[R_{\M_\A,f}|w\models \psi]\ge \val_{\M'_\A}(g)$. By Lemma~\ref{lem:M to M'}, the strategy $f$ maximizes $\EE[R_{\M_\A,f}|w\models \psi]$, as otherwise $g$ does not attain the maximal value in $\M'_\A$.
Thus, it is enough to find a maximizing memoryless strategy in $\M'_\A$, which can be done in doubly-exponential time.

The lower bound trivially follows from Theorem~\ref{thm:M to T shqsyn}.

\qed

\subsection{Proof of Lemma~\ref{lem: conditional threshold}}
\label{apx:conditional threshold}
By the law of total probability, we can write
\begin{align*}
&\Pr(\sem{\T(w),\psi\to \phi}\ge t)=\\
&=\Pr(\sem{\T(w),\psi\to \phi}\ge t|w\models \psi) \cdot \Pr(w\models\psi)+\\
&+\Pr(\sem{\T(w),\psi\to \phi}\ge t|w\models \neg\psi) \cdot (1-\Pr(w\models\psi))
\end{align*}
Given that $\psi$ holds, we have that $\sem{\T(w),\psi\to \phi}=\sem{\T(w),\phi}$. Therefore, 
$\Pr(\sem{\T(w),\psi\to \phi}\ge t|w\models \psi)=\Pr(\sem{\T(w),\phi}\ge t)$.

Given that $\psi$ does not hold, we have that $\sem{\T(w),\psi\to \phi}=1$. Therefore, $\Pr(\sem{\T(w),\psi\to \phi}\ge t|w\models \neg\psi)=1$.

Accordingly, $\Pr(\sem{\T(w),\psi\to \phi}\ge t)=1$ iff 
$\Pr(\sem{\T(w),\phi}\ge t) \cdot \Pr(w\models\psi)+1\cdot(1-\Pr(w\models\psi))=1$. Assuming $0<\Pr(\psi)<1$, the latter is equivalent to  $\Pr(\sem{\T(w),\phi}\ge t)=1$, and we are done.
\qed

\end{document}